\documentclass{lmcs} 
\pdfoutput=1

\usepackage{lastpage}
\renewcommand{\dOi}{14(4:8)2018}
\lmcsheading{}{1--\pageref{LastPage}}{}{}%
{Dec.~23,~2016}{Oct.~30,~2018}{}


\usepackage{hyperref}
\usepackage{amsmath,amsthm,amssymb}
\usepackage{braket}
\usepackage{proof}
\usepackage{tikz-cd}
\usepackage{url}
\usepackage{hyperref}

\renewcommand{\leq}{\leqslant}
\renewcommand{\geq}{\geqslant}
\renewcommand{\ngeq}{\ngeqslant}
\renewcommand{\nleq}{\nleqslant}
\renewcommand{\phi}{\varphi}
\newcommand{\pfin}{\mathcal{P}_\mathrm{fin}}

\theoremstyle{plain} 

\newtheorem*{ct}{Completeness theorem (CT)}

\begin{document}

\title[Extension by Conservation. Sikorski's Theorem]{Extension by Conservation. Sikorski's Theorem}

\author[D.~Rinaldi]{Davide Rinaldi}	
\author[D.~Wessel]{Daniel Wessel}	
\address{Universit\`a di Verona, Dipartimento di Informatica,
Strada le Grazie 15, 37134 Verona, Italy}	
\email{daviderinaldimath@gmail.com}  
\email{daniel.wessel@univr.it}  







\begin{abstract}
  \noindent Constructive meaning is given to the assertion that every finite Boolean algebra
  is an injective object in the category of distributive lattices.
  To this end, we employ Scott's notion of entailment relation,
  in which context we describe Sikorski's extension theorem for finite Boolean algebras
  and turn it into a syntactical conservation result.
  As a by-product, we can facilitate proofs of related classical principles.
\end{abstract}

\maketitle

\section{Introduction}

Due to a time-honoured result by Sikorski (see, e.g., \cite[\S 33]{sikorski:ba} and \cite{halmos:loba}),
the injective objects in the category of Boolean algebras have been identified precisely as the complete Boolean algebras.
In other words, a Boolean algebra $C$ is complete if and only if,
for every morphism $f : B \rightarrow C$ of Boolean algebras, where $B$ is a subalgebra of $B'$,
there is an extension $g : B' \rightarrow C$ of $f$ onto $B'$.
%
%
More generally, it has later been shown by Balbes \cite{balbes:proin}, and Banaschewski and Bruns \cite{banaschewski:inhulls},
that a distributive lattice is an injective object in the category of distributive lattices
if and only if it is a complete Boolean algebra.

A popular proof of Sikorski's theorem proceeds as follows:
by Zorn's lemma the given morphism on $B$ has a maximal partial extension,
which by a clever one-step extension principle \cite{halmos:loba,bell:pers}
can be shown to be total on $B'$.
In turn, instantiating $C$ with the initial Boolean algebra $\mathbf{2}$ with two elements
results in a classical equivalent of the Boolean prime ideal theorem,
a proper form of the axiom of choice \cite{rubin:equivalents}.
Taking up a strictly constructive stance,
it is even necessary to object to completeness of $\mathbf{2}$,
as this implies the principle of weak excluded middle \cite{bell:ist}.
We recall further that while in classical set theory Sikorski's theorem is stronger than the Boolean prime ideal theorem \cite{bell:strength},
the latter principle is in fact equivalent to the restricted form of Sikorski's theorem
for complete and atomic Boolean algebras \cite{luxemburg:sik}.

We can give constructive meaning to Sikorski's extension theorem for finite Boolean algebras
by phrasing it as a syntactical conservation result.
The idea is as follows.
Given a distributive lattice $L$
and a finite discrete Boolean algebra $B$,
we generate an \emph{entailment relation} \cite{scott:engender,sco:com,sco:bac}
the models of which are precisely the lattice maps $L \rightarrow B$
(preliminaries will be laid out in Section \ref{sec:er}).
This can be done in a canonical manner,
and in particular so with every lattice $L'$ containing $L$ as a sublattice.
Then we have two entailment relations $\vdash$ and $\vdash'$,
the former of which can be \emph{interpreted} in the latter.
By way of a suitable form of the axiom of choice, with
some classical logic,
this interpretation being \emph{conservative} is tantamount to the
restriction of lattice maps $L' \rightarrow B$ to the sublattice $L$ being surjective---which is extendability.
The proof of conservativity, however,
is elementary and constructive,
and it rests upon the explicit characterization of $\vdash$
in terms of an appropriate \emph{Formal Nullstellensatz} (see Section \ref{sec:nsts}).

We hasten to add that this approach does not stem from an altogether new idea.
It is quite in order to cite Mulvey and Pelletier \cite[p. 41]{mulvey:globhb},
who grasp the ``\emph{intuitive content which the Hahn-Banach theorem normally brings to functional analysis}'' in view of that

\begin{quote}
``[...] \emph{no more may be proved about the subspace $A$
in terms of functionals on the seminormed space $B$
than may already be proved by considering only functionals on the subspace $A$.}''
\end{quote}

In this spirit, the Hahn-Banach theorem has been revisited by way of type theory \cite{negri:hbintt},
as well as in terms of entailment relations \cite{coquand:erdl,coquand:geomhb,coquand:dlhb}.
Further approaches to extension theorems in the guise of logical conservation
include Szpilrajn's order extension principle \cite{negri:ptorder}.
The idea of capturing algebraic structures
and in particular their ideal objects
by way of entailment relations
has been developed and advocated in \cite{coquand:erdl,coquand:dlhb,coquand:vdp}.
The interplay of (multi-conclusion) entailment relations as
extending their single-conclusion counterpart is subject of \cite{rsw:edde:abstract,rsw:edde},
of course building on \cite{sco:com}.
Clearly, we draw inspiration from these references.

\paragraph{On method and foundations}
If not explicitly mentioned otherwise,
throughout we reason constructively.
All what follows may be formalized in a suitable fragment of constructive Zermelo-Fraenkel set theory ($\mathbf{CZF}$)
\cite{aczel:notes,aczel:cstdraft}.
However, in order to bridge the gap towards the classical counterpart of our conservation result,
we need to employ a version of the completeness theorem ($\mathbf{CT}$) for entailment relations,
which in fact is a form of the axiom of choice,
and to invoke the law of excluded middle to point out some well-known classical consequences.
In these cases we refer to classical set theory ($\mathbf{ZFC}$)
and this will be indicated appropriately.
Recall that a set $S$ is said to be \emph{discrete} if equality on the set is decidable, which is to say that
$\forall x, y \in S\, (\, x = y \lor x \neq y\, )$,
and a subset $U$ of $S$ is called \emph{detachable} if
$\forall x \in S\, (\, x \in U \lor x \notin U\, )$.
A set $S$ is \emph{finite} if there exists $n \in \mathbb{N}$
and a surjective function $f : \set{1, \dots, n} \rightarrow S$.
The class of finite subsets of $S$ is denoted $\mathcal{P}_\mathrm{fin}(S)$,
which forms a set in $\mathbf{CZF}$.
Again, we refer to \cite{aczel:notes,aczel:cstdraft}.

\section{Entailment relations}\label{sec:er}

Let $S$ be a set.
An \emph{entailment relation} \cite{sco:com,scott:engender,sco:bac} on $S$ is a relation
\[
\vdash\ \subseteq \pfin(S) \times \pfin(S)
\]
between \emph{finite} subsets of $S$,
which is \emph{reflexive}, \emph{monotone}, and \emph{transitive}
in the following sense,
written in rule notation:
%
\[
\infer[(\mathrm{R})]{X \vdash Y}{X \between Y}
\qquad\qquad
\infer[(\mathrm{M})]{X, X' \vdash Y, Y'}{ X \vdash Y}
\qquad\qquad
\infer[(\mathrm{T})]{X \vdash Y}{X \vdash Y,x \quad X,x \vdash Y}
\]
Transitivity (T) is an abstract form of Gentzen's \emph{cut} rule for sequent calculus.
In (R) the notation $X \between Y$ means that $X$ and $Y$ have an element in common.\footnote{We
have borrowed this notation from formal topology \cite{sambin:somepoints}.}
We write $X,Y$ rather than $X \cup Y$,
as well as just $x$ where it actually should read a singleton set $\set{x}$,
but this is a matter of readability only.
We refer to the elements of $S$ as \emph{abstract statements};
one might even call them formulae,
even though $S$ need not consist of syntactic objects per se,
that is to say, formulae in the sense specified by a certain formal language.
As regards intuition, entailment is to be read just as a Gentzen sequent:
conjunctively on the left, and disjunctively
on the right hand side of the turnstile symbol $\vdash$.

Here is an important and recurring example.
If $L$ is a distributive lattice,
then a natural choice for an entailment relation on $L$,
where every element is considered to be an abstract statement,
is given by
\[
X \vdash Y
\qquad
\text{if and only if}
\qquad
\bigwedge X \leq \bigvee Y.
\]
While reflexivity and monotonicity trivially hold,
the simple proof that $\vdash$ is transitive
rests on distributivity of $L$ (in fact, cut is equivalent to distributivity \cite{sco:com}),
see, e.g., \cite[XI, \S4.5]{lombardiquitte:constructive}.

Notice that the definition of entailment relation is symmetric:
if $\vdash$ is an entailment relation, then so is its converse $\dashv$.
The \emph{trivial} entailment relation on a set $S$ is the one for which $X \vdash Y$
for \emph{every} pair $X,Y$ of finite subsets of $S$.
Besides, apart from being trivial,
this is the largest entailment relation on $S$.
If there is at least one pair $X,Y$ between which entailment cannot be inferred,
then we say that $\vdash$ is \emph{non-trivial}.

Given a set $\mathcal{E}$ of pairs of finite subsets of $S$,
we may consider the entailment relation \emph{generated by} $\mathcal{E}$.
This is the closure of $\mathcal{E}$ under the rules (R), (M) and (T),
and thus is the least entailment relation on $S$ to contain $\mathcal{E}$.
An important point about generated entailment relations concerns proof technique.
For instance,
if we want to show that all instances of an entailment relation $\vdash$ share a certain property,
then we may argue inductively,
showing the property in question first to hold for axioms,
then going through the rules,
which allows to apply the hypothesis accordingly.



\begin{defi}
Let $\vdash$ be an entailment relation on a set $S$.
An \emph{ideal element} (or \emph{model}, \emph{point}) \cite{coquand:erdl,coquand:sfc}
of $\vdash$ is a subset $\alpha \subseteq S$ which splits entailment,
i.e.,
\[
\infer{\alpha \between Y}{\alpha \supseteq X \quad X \vdash Y}
\]
\end{defi}

We write $\mathfrak{Spec}(\vdash)$ for the collection of ideal elements for $\vdash$.
For instance, overlap $\between$
is an entailment relation on $S$ if restricted to hold between finite subsets of $S$ only,
in which case the ideal elements are arbitrary subsets,
thus $\mathfrak{Spec}(\between) = \mathcal{P}(S)$.\footnote{Hence
$\mathfrak{Spec}(\vdash)$ might well be a proper class,
as $\mathbf{CZF}$ does not employ the axiom of powerset!
However, neither will this be an issue, nor does it pose a problem:
$\mathbf{CZF}$ covers the axiom of exponentiation \cite{aczel:cstdraft},
and for what we have in mind, this suffices.
Moreover, once we are in need of completeness (see below),
we will find it necessary to switch to a classical setting, anyway.}

We say that a finite set $X$ of statements is \emph{inconsistent}
if $X \vdash$.
By monotonicity, if $X$ is inconsistent,
then $X \vdash Y$ for every finite subset $Y$ of $S$.
Notice that an ideal element cannot have an inconsistent subset.

Often, working towards a non-inductive description of an inductively
defined entailment relation, i.e., to what is sometimes called a \emph{formal Nullstellensatz}
for the entailment relation at hand,
it is preferable to concentrate first on desribing its empty-conclusion instances.
Lemma \ref{lem:incpred} below,
which is a consequence of cut elimination for entailment relations \cite{rinaldiwessel:cut},
can be helpful to this end;
we will give an \emph{ad hoc} argument.
In the following,
we identify subsets $\Phi$ of $\mathcal{P}_\mathrm{fin}(S)$ with their defining predicate, i.e.,
we write $\Phi(X)$ as a shorthand to indicate membership $X \in \Phi$.
We say that $\Phi$ is (upward) \emph{monotone} if $\Phi(X)$ and $X \subseteq Y$ imply $\Phi(Y)$,
where of course $X,Y \in \mathcal{P}_\mathrm{fin}(S)$.
\begin{lem}\label{lem:incpred}
Let $\Phi$ be a family of finite inconsistent subsets of $S$.
The following are equivalent.
\begin{enumerate}
  \item $\Phi \supseteq \Set{ X \in \mathcal{P}_\mathrm{fin}(S) : X \vdash }$.
  \item\label{incpred:resolve}
  The rule
  \[
  \infer{\Phi(x_1, \dots, x_k,Z)}{x_1, \dots, x_k \vdash y_1, \dots, y_\ell \qquad \Phi(y_1,Z) \quad \dots \quad \Phi(y_\ell,Z)}
  \]
  is provable.
\end{enumerate}
If $\Phi$ is monotone and $\vdash$ is inductively defined,
then it suffices in (\ref{incpred:resolve}) to consider generating entailments only.
\end{lem}
\begin{proof}
Inferring the second item from the first is merely a matter of transitivity (T).
For the converse, i.e.,
to obtain $\Phi(X)$ from $X \vdash$ given (\ref{incpred:resolve}),
simply instantiate the latter with $Z = \emptyset$.
As regards the add-on, if $\vdash$ is inductively generated from a set $\mathcal{E}$ of initial entailments,
then by an inductive argument
\[
\infer{\Phi(X,Z)}{(X,Y)\in \mathcal{E} \quad \forall y \in Y\,\Phi(y,Z)}
\]
is readily shown equivalent to (\ref{incpred:resolve}).
\end{proof}

If $\vdash$ and $\vdash'$ are two entailment relations on sets $S$ and $S'$, respectively,
then an \emph{interpretation} \cite{coquand:topandsc,coquand:sfc} of the former is a function $f : S \rightarrow S'$ such that
$X \vdash Y$ implies $f(X) \vdash' f(Y)$.
An interpretation is said to be \emph{conservative} if on the other hand
$f(X) \vdash' f(Y)$ implies $X \vdash Y$.
Every interpretation $f : (S, \vdash) \rightarrow (S', \vdash')$ induces a mapping
\[
f^{-1} : \mathfrak{Spec}(\vdash') \rightarrow \mathfrak{Spec}(\vdash), \quad \alpha \mapsto f^{-1}(\alpha)
\]
of ideal elements.
In fact, the property of $f$ being an interpretation ensures that the inverse image mapping
$f^{-1} : \mathcal{P}(S) \rightarrow \mathcal{P}(S)$
restricts on ideal elements.

According to the \emph{Fundamental Theorem of Entailment Relations} \cite[Theorem 3]{coquand:erdl}
every entailment relation generates a distributive lattice
with conservative universal interpretation.
That is, if $(S, \vdash)$ is a set together with an entailment relation $\vdash$,
then there is a distributive lattice $L(S)$,
together with a map $i : S \rightarrow L(S)$
such that
\[
X \vdash Y
\qquad
\textrm{if and only if}
\qquad
\bigwedge_{x \in X} i(x) \leq \bigvee_{y \in Y} i(y), \tag{$\ast$}
\]
and whenever $f : S \rightarrow L'$ is an interpretation of $\vdash$ in another distributive lattice $L'$,
i.e., satisfying the left-to-right implication of $(\ast)$ with $f$ in place of $i$,
then there is a lattice map $f' : L(S) \rightarrow L'$ such that the following diagram commutes
\[
\begin{tikzcd}
S \arrow{r}{i} \arrow{rd}[swap]{f}		&		L(S) \arrow[dashed]{d}{f'}	\\
~				&		L'
\end{tikzcd}
\]
This has several important applications,
and provides connections with point-free topology~\cite{coquand:erdl}.
Generated lattices can further be used \cite{coquand:erdl,coquand:sfc} to verify the

\begin{ct}
The following are equivalent.
\begin{enumerate}

\item
$X \vdash Y$

\item
$\forall \alpha \in \mathfrak{Spec}(\vdash)\ (\ \alpha \supseteq X\ \rightarrow\ \alpha \between Y\ )$

\end{enumerate}
%
\end{ct}

The completeness theorem for entailment relations
is implied by spatiality of coherent frames \cite{coquand:sfc},
hence is a form of the axiom of choice.
Completeness will be used to derive classical consequences
from the formal Nullstellensatz (Theorem \ref{hilbert}).

We end this section with an important consequence of completeness,
crucial for our purpose.
It is taken from \cite{coquand:topandsc}.

\begin{thm}[``Conservativity = Surjectivity'', \textbf{ZFC}]\label{thm:conssurj}
If $f: (S, \vdash) \rightarrow (S', \vdash')$ is an interpretation of entailment relations,
then $f$ is conservative if and only if the induced mapping $f^{-1}$ is surjective on ideal elements.
\end{thm}

\section{Sikorski by entailment}

\subsection{On lattices and Boolean algebras}\label{onlattices}

In the following,
every lattice $L$ is considered to be distributive, and bounded,
i.e., to have a top and bottom element, $1_L$ and $0_L$, respectively.
Subscripts will be written in order to emphasize to which lattice we refer,
otherwise they will be omitted. 
We understand $1$ to be the empty meet
and $0$ to be the empty join in $L$.
Maps between lattices are meant to preserve structure.
Mind that a lattice is discrete if and only if its order relation $\leq$ is decidable.
An \emph{atom} of a lattice $L$ is an element $e \in L$ which is minimal among non-zero elements,
i.e., for every $x \in L$, if $0 < x \leq e$, then $x = e$;
of course $x < y$ is shorthand for $x \leq y$ and $x \neq y$.
The set of all atoms of $L$ will be denoted $\mathrm{At}\, L$.

We will be dealing with finite discrete Boolean algebras only,
for which there is the following \emph{Structure Theorem}
\cite[VII, \S 3, 3.3 Theorem]{lombardiquitte:constructive}:

\begin{thm}
Every finite discrete Boolean algebra is isomorphic to the algebra of the detachable
subsets of a finite discrete set.
\end{thm}

Crucially,
every finite discrete Boolean algebra $B$ has a finite set of atoms,
and $1_B = \bigvee \mathrm{At}\, B$.
It follows that every element of $B$ is a finite join of atoms.
We refer to \cite[VII, \S 3]{lombardiquitte:constructive}.

There are several classically equivalent ways of describing atoms in a Boolean algebra \cite{handbook:BA1}.
They coincide, however, under assumption of discreteness.
%
%
%
In particular, for every $e \in B$ the following are equivalent \cite{handbook:BA1},\cite[VII, \S 3]{lombardiquitte:constructive}.
\begin{enumerate}
\item
$e \in \mathrm{At}\, B$.

\item
$e > 0$ and, for every $a \in B$, either $e \leq a$ or $e \wedge a = 0$.

\item
$e > 0$ and, for every $a \in B$, either $e \leq a$ or $e \leq -a$.
\end{enumerate}

Notice further that if $U \subseteq \mathrm{At}\, B$ is finite,
and $e \in \mathrm{At}\, B$,
then $e \leq \bigvee U$ implies $e \in U$.
Moreover, since every element of $B$ can be written as a finite join of atoms,
if $b \nleq b'$, then there is $e \in \mathrm{At}\, B$ such that
$e \leq b \wedge -b'$.

\subsection{Entailments for maps}\label{sequentsformaps}

%
%

Now we fix a distributive lattice $L$
and a finite discrete Boolean algebra $B$,
according to the conventions in the preceding Section \ref{onlattices}.
Confident that it will not lead to confusion,
the order relations on $L$ and $B$
will both be denoted by $\leq$.
As our domain of discourse we take $L \times B$,
and study an entailment relation $\vdash$ on $L \times B$,
inductively generated by the set of all instances of the following axioms.
\begin{align*}
(x, a), (x, b) &\vdash \tag{s}\\
(x, a), (y, b) &\vdash (x \wedge y, a \wedge b) \tag{$\wedge$}\\
(x, a), (y, b) &\vdash (x \vee y, a \vee b) \tag{$\vee$}\\
&\vdash (0_L, 0_B) \tag{$0$}\\
&\vdash (1_L, 1_B) \tag{$1$}\\
&\vdash \Set{ (x,a) : a \in B} \tag{t}
\end{align*}
where $a \neq b$ in (s).

We follow an idea outlined in \cite{coquand:topandsc},
where suitable axioms of the kind (t) and (s) are taken to present
the space of functions $\mathbb{N} \rightarrow \set{0,1}$
by way of a generated entailment relation.

We read any pair $(x,a)$
as $\phi(x) = a$ for a generic (or yet-to-be-determined)
morphism $\phi: L \rightarrow B$ of lattices.
In this regard, entailment
\[
(x_1, a_1), \dots, (x_k, a_k) \vdash (y_1, b_1), \dots, (y_\ell, b_\ell)
\]
should intuitively be read as
\[
\textit{if } \phi(x_1) = a_1 \dots \textit{and} \dots \phi(x_k) = a_k,
\textit{ then }
\phi(y_1) = b_1 \dots \textit{or} \dots \phi(y_\ell) = b_\ell.
\]
Notice that an ideal element $\alpha \subseteq L \times B$ for $\vdash$ is a relation,
in the first place.
Axiom (t) forces such an $\alpha$ to be total,
while the second axiom (s) ensures single values.
The remaining axioms are to
guarantee that the lattice structure is preserved.
We put on record that this entailment relation really describes what we intend it to:

\begin{lem}
A subset $\alpha \subseteq L \times B$ is an ideal element of $\vdash$
if and only if it is a homomorphism of lattices.
\end{lem}

Entailments are in good accordance with our intuition about lattice maps.
For instance, since every lattice map is order-preserving,
we should expect that an entailment like $X \vdash (x,a)$
sets a certain bound on the set of those abstract statements $(y, b)$,
which still are consistent with $X$
in case $x \leq y$.

\begin{lem}\label{lem:nmonotonic}
For all $x,y \in L$ and $a,b \in B$,
if $x \leq y$ and $a \nleq b$,
then $(x, a), (y,b) \vdash $.
\end{lem}

\begin{proof}
Notice that we have
\[
(x,a), (y,b) \vdash (x, a \wedge b)
\]
by $(\wedge)$ and since $x \wedge y = x$.
This entailment can be cut with
\[
(x,a), (x, a \wedge b) \vdash
\]
which we have as an instance of (s) because $a \neq a \wedge b$.
\end{proof}

%

It will be useful to have means for moving statements back and forth across the turnstile symbol.
The idea is as follows.
If a set $X$ of statements logically forces an element $x \in L$ to take a certain value $a \in B$
under every given lattice map $L \rightarrow B$,
then adjoining some statement $(x,b)$ to $X$
should turn out inconsistent in case $b \neq a$.

\begin{lem}\label{lem:hinher}
For every finite subset $X \subseteq L \times B$
and elements $x \in L, a \in B$,
the following are equivalent.
\begin{enumerate}

\item
$X \vdash (x, a), Y$

\item
$X, (x, b) \vdash Y$
for every $b \in B$ such that $b \neq a$.

\end{enumerate}
\end{lem}

\begin{proof}
Suppose that $X \vdash (x,a), Y$ is inferrable.
If $b \neq a$,
then $(x, a), (x, b) \vdash$ is an axiom,
whence we get $X, (x, b) \vdash Y$ by cut.
On the other hand, if $X, (x, b) \vdash Y$ whenever $b \neq a$,
then we can successively cut $\vdash \Set{ (x, b) : b \in B}$,
until we arrive at $X \vdash (x, a), Y$.
\end{proof}

\begin{exa}\label{hinher2}
Lemma \ref{lem:hinher} has a particularly simple form for the Boolean algebra $\mathbf{2} = \set{0,1}$,
in which case it amounts to complementation of values:
for every $x \in L$ and $i \in \mathbf{2}$ we have $(x,i) \vdash$ if and only if $\vdash (x,-i)$.
\end{exa}

\subsection{A formal Nullstellensatz}\label{sec:nsts}


Following the tradition of dynamical algebra \cite{lombardi:krull,cos:dyn,coq:hidden-krull},
a \emph{formal Nullstellensatz}\footnote{Hilbert's Nullstellensatz as a proof-theoretic result
appears first in \cite{scarp:meta}, see also \cite{lifschitz:sct}.}
for $\vdash$
is a theorem explicitly describing entailment
in terms of certain identities within the algebraic structure at hand.
We will concentrate on a weak form first,
which is to say that we give a direct description of inconsistency.
In view of Lemma \ref{lem:hinher},
this will indeed suffice for a direct, non-inductive description of $\vdash$ overall.


We find it useful to have an additional simple piece of notation.
If $X$ is a finite subset of $L \times B$ and $a \in B$,
then let
\[
X_a = \Set{ x \in L : (x, a) \in X }.
\]
Note that $X_a$ is finite as well.

The following three lemmas are technical and provide for the proof of Theorem \ref{hilbert} below.
First we give a condition which is sufficient for finite sets of statements to be inconsistent.

\begin{lem}\label{contra}
Let $X$ be a finite subset of $L \times B$.
If there is $e \in \mathrm{At}\, B$ such that
\[
\bigwedge_{a \geq e} \bigwedge X_a
\leq
\bigvee_{a \geq e} \bigvee X_{- a},
\]
then $X$ is inconsistent, i.e., $X \vdash$.
\end{lem}

\begin{proof}
We can write $X = \set{ (x_1, a_1), \dots, (x_k, a_k) }$.
Suppose that there is an atom $e$ as indicated.
Let $a \in \set{a_1, \dots, a_k}$.
For every $x \in X_a$
we have $X \vdash ( x, a )$,
and since
\[
\Set{ (x,a) : x \in X_a } \vdash \big( \bigwedge X_a, a \big)
\]
is inferrable, we get
\[
X \vdash \big(\bigwedge X_a, a\big)
\]
by cut.
In particular, we have such an entailment whenever $a \in \set{a_1, \dots, a_k}$ and $a \geq e$.
Thus, writing
\[
x = \bigwedge_{\substack{ a \geq e \\ a \in \set{a_1, \dots, a_k}} } \bigwedge X_a
\]
and
\[
b = \bigwedge \Set{ a \in \set{a_1, \dots, a_k} : a \geq e },
\]
we are able to infer
\[
X \vdash ( x, b ).
\]
In a similar manner, for
\[
y = \bigvee_{\substack{ a \geq e \\ a \in \set{-a_1, \dots, -a_k} }} \bigvee X_{-a}
\]
and
\[
b' = \bigvee \Set{ -a : a \in \set{ -a_1, \dots, -a_k} \text{ and } a \geq e }
\]
we are able to infer
\[
X \vdash (y, b').
\]
It remains to notice that on the one hand we actually have
\[
x
=\bigwedge_{a \geq e} \bigwedge X_a
\qquad
\text{and}
\qquad
y = \bigvee_{a \geq e} \bigvee X_{- a}.
\]
Thus $x \leq y$, according to the assumption.
On the other hand,
we have
\[
e \leq b
\qquad
\text{and}
\qquad
b' \leq -e
\]
which implies $b \nleq b'$.
Therefore, by way of Lemma \ref{lem:nmonotonic} we get
\[
(x,b), (y,b') \vdash
\]
and now $X \vdash$ is immediate.
\end{proof}

We are going to make use of the following combinatorial principle.

\begin{lem}\label{lem:combinatorialcut}
Let $A$ be a finite inhabited set and let $L$ and $R$ be subsets of $A$.
If, for every finite subset $U$ of $A$, there is an element $a \in A$ such that
\[
(\,a \in U\, \land\, a \in L\,)\, \lor\, (\,a \notin U\, \land\, a \in R\,) \tag{$\dagger$}
\]
then $L$ and $R$ have an element in common, $L \between R$.
\end{lem}


\begin{proof}
We argue by induction on the finite number of elements of $A$.
First we consider a singleton set $A = \set{\ast}$
under assumption of $(\dagger)$.
If we instantiate with $U = \emptyset \subseteq A$,
then this yields $\ast \in R$ immediately.
If instead we instantiate with $U = \set{\ast}$,
then we are led to the left-hand disjunct, thus $\ast \in L$.

Next we consider $A' = A \cup \set{\ast}$,
where $\ast$ is an element not among those of $A$.
We suppose that the principle in question is valid for $A$,
and that $(\dagger)$ applies with respect to $A'$.
In particular,
for every finite subset $U$ of $A$,
there is $a \in A'$ such that
\[
(\,a \in U\, \land\, a \in L\,)\, \lor\, (\,a \notin U\, \land\, a \in R\,).
\]
Since either $a \in A$ or $a = \ast$
we get
\[
(\,a \in U\, \land\, a \in L\cap A\,)\, \lor\, (\,a \notin U\, \land\, a \in R\cap A\,)\, \lor\, (\, \ast \in R\,).
\]
Then the inductive hypothesis applies,
whence we infer $(L\cap A) \between (R \cap A)$ or $\ast \in R$.
Similarly, for every finite subset $U$ of $A$ there is $a \in A'$ such that
\[
(\,a \in U \cup \set{\ast}\, \land\, a \in L\,)\, \lor\, (\,a \notin U\cup \set{\ast}\, \land\, a \in R\,).
\]
Again, since either $a \in A$ or $a = \ast$, we get
\[
(\,a \in U\, \land\, a \in L\cap A\,)\, \lor\, (\,a \notin U\, \land\, a \in R\cap A\,)\, \lor\, (\, \ast \in L\,)
\]
which with the inductive hypothesis leads to $(L\cap A) \between (R \cap A)$ or $\ast \in L$.
Taken together, this yields
\[
(L\cap A) \between (R \cap A)\, \lor\, (\,\ast \in L \cap R\,),
\]
whence we have $L \between R$.
%
%
%
%
%
%
%
%
\end{proof}

\begin{rem}
Classically, we could have given a much shorter proof of Lemma \ref{lem:combinatorialcut}.
In the classical setting, if $A$ is finite, then the subset $R\subseteq A$ has to be finite itself.
We can then instantiate ($\dagger$) by $U=R$, from which the result trivially follows.
\end{rem}

The following may be considered a cut rule for inconsistent sets of statements.

\begin{lem}\label{atomiccut}
Let $X$ be a finite subset of $L\times B$ and let $x \in L$.
If for every $b \in B$ there is $e \in \mathrm{At}\, B$ such that
\[
\bigwedge_{a \geq e} \bigwedge \big(X,(x, b)\big)_a
\leq
\bigvee_{a \geq e} \bigvee \big(X,(x,b)\big)_{- a}
\]
then there is $e' \in \mathrm{At}\, B$ such that
\[
\bigwedge_{a \geq e'} \bigwedge X_a
\leq
\bigvee_{a \geq e'} \bigvee X_{- a}
\]
\end{lem}

\begin{proof}
Notice first that if $e$ is an atom of $B$, then the inequality
\[
\bigwedge_{a \geq e} \bigwedge \big(X,(x, b)\big)_a
\leq
\bigvee_{a \geq e} \bigvee \big(X,(x,b)\big)_{- a} \tag{$\#$}
\]
amounts to
\[
\bigwedge_{a \geq e} \bigwedge X_a \wedge x
\leq
\bigvee_{a \geq e} \bigvee X_{- a} \tag{$L_e$}
\]
in case $b \geq e$,
and to
\[
\bigwedge_{a \geq e} \bigwedge X_a
\leq
\bigvee_{a \geq e} \bigvee X_{- a} \vee x \tag{$R_e$}
\]
otherwise, i.e., in case of $b \ngeq e$.
We need to find an atom $e$ for which both $L_e$ and $R_e$ hold---cut
in the lattice $L$ then allows to draw the desired conclusion.
To this end, set
\[
L = \Set{ e \in \mathrm{At}\,B : L_e }
\qquad
\text{and}
\qquad
R = \Set{ e \in \mathrm{At}\,B : R_e }.
\]
In particular, for every finite subset $U$ of $\mathrm{At}\, B$
our assumption applies to the finite join $\bigvee U$
for which there is $e \in \mathrm{At}\,B$ such that either
$e \leq \bigvee U$ and $L_e$,
or else $e \nleq \bigvee U$ and $R_e$.
Taking into account that $e \leq \bigvee U$ if and only if $e \in U$,
we see that the combinatorial Lemma \ref{lem:combinatorialcut} applies,
whence we get $L \between R$.
\end{proof}


%

Finally, here is how to describe the finite inconsistent subsets for $\vdash$ explicitly.
As it turns out, entailment $X \vdash $ amounts to a certain inequality in the lattice $L$.
This (weak) formal Nullstellensatz lies at the heart of conservation.
\begin{thm}[Weak Nullstellensatz]\label{hilbert:incons}
For every finite subset $X$ of $L \times B$,
the following are equivalent.
\begin{enumerate}
\item
$X \vdash$
\item
There is an atom $e \in \mathrm{At}\, B$ such that
\[
\bigwedge_{a \geq e} \bigwedge X_a \leq \bigvee_{a \geq e} \bigvee X_{- a}
\]
\end{enumerate}
\end{thm}
\begin{proof}
We apply Lemma \ref{lem:incpred} to which end we consider the family $\Phi$ of all finite subsets of $L \times B$
satisfying the second item of Theorem \ref{hilbert:incons}.
Apparently, this $\Phi$ is monotone.
Lemma \ref{contra} asserts that every member of $\Phi$ is inconsistent.
In order to show $\Phi = \Set{ X \in \mathcal{P}_\mathrm{fin}(S) : X \vdash }$,
it thus suffices to check item (\ref{incpred:resolve}) of Lemma \ref{lem:incpred}
with respect to every initial entailment.
But this is straightforward except for the axiom of totality (t),
which is taken care of by Lemma~\ref{atomiccut}.
\end{proof}
\begin{cor}[Formal Nullstellensatz]\label{hilbert}
For every finite subset $X$ of $L \times B$,
and every finite set of pairs $(y_1, b_1), \dots, (y_k, b_k) \in L \times B$,
the following are equivalent.
\begin{enumerate}
\item
$X \vdash (y_1, b_1), \dots, (y_k, b_k)$
\item
For all $b_1' \neq b_1, \dots, b_k' \neq b_k$
there is $e \in \mathrm{At}\, B$ such that
\[
\bigwedge_{a \geq e} \bigwedge \big( X, \set{(y_i, b_i')}_{1\leq i\leq k} \big)_a
\leq
\bigvee_{a \geq e} \bigvee \big( X, \set{(y_i, b_i')}_{1 \leq i \leq k} \big)_{- a}
\]
\end{enumerate}
\end{cor}
\begin{proof}
Apply Lemma \ref{lem:hinher} repeatedly in order to reduce to empty conclusion entailments.
Theorem \ref{hilbert:incons} then yields the claim.
\end{proof}

It is interesting to note that
non-triviality of $\vdash$ is for free,
given that $L$ is non-trivial:
\begin{cor}
The following are equivalent.
\begin{enumerate}
\item
$\emptyset \vdash \emptyset$
\item
$1 = 0$ in $L$.
\end{enumerate}
%
\end{cor}
\begin{proof}
Since we have $\vdash (0_L, 0_B)$ and $\vdash (1_L, 1_B)$ as axioms,
the entailment relation $\vdash$ is trivial if and only if $(0_L, 0_B), (1_L, 1_B) \vdash$ can be inferred.
But if $e$ is an arbitrary atom of $B$, then
\[
1_L = \bigwedge_{a \geq e} \bigwedge \set{ (0_L, 0_B), (1_L, 1_B) }_a
\leq
\bigvee_{a \geq e} \bigvee \set{ (0_L, 0_B), (1_L, 1_B) }_{-a} = 0_L.
\]
On the other hand,
if indeed $1 = 0$ in $L$,
then $\emptyset \vdash \emptyset$ can be inferred accordingly.
\end{proof}

How to formally translate Theorem \ref{hilbert:incons} into
the classical extension theorem will be explained in the next section.
First let us see a couple of interesting consequences
that can be derived from the formal Nullstellensatz.



\begin{exa}\label{exa:orderreflection}
The order relation on $L$ can be expressed in terms of entailment.
With the formal Nullstellensatz Corollary \ref{hilbert} it can be shown that for all $x,y \in L$ the following are equivalent:
\begin{enumerate}
\item
$x \leq y$
\item
$(x, 1_B) \vdash (y, 1_B)$
\item
$(y, 0_B) \vdash (x, 0_B)$
\end{enumerate}

\end{exa}

\begin{exa}\label{nullenundeinsen}
For every $x \in L$ the following are equivalent:
\begin{enumerate}
\item
$x = 0_L$
\item
$\vdash (x,0_B)$
\item
$(x,1_B) \vdash$
\end{enumerate}
With completeness,
this leads over to the assertion that
$0 \in L$ is the only element which maps to $0 \in B$
under \emph{every} lattice map $\phi : L \rightarrow B$.
%

Dually, for every $x \in L$ the following are equivalent:
\begin{enumerate}
\item
$x = 1_L$
\item
$\vdash (x,1_B)$
\item
$(x,0_B) \vdash$
\end{enumerate}
\end{exa}

\begin{rem}
Results about entailment relations
can sometimes be used to facilitate
proofs of classical theorems.
This advantage has also been pointed out in \cite{coq:hidden-krull},
and includes the version of Sikorski's theorem considered in this paper.
Of course, completeness (CT) needs to be invoked to this end.
We briefly mention another example,
which leads to the \emph{Representation Theorem} for distributive lattices \cite[Prop.~I.2.5]{johnstone:stsp},
a classical equivalent of CT:
\begin{quote}\emph{
If $L$ is a distributive lattice and $x,y \in L$ are such that $x \nleq y$,
then there exists a homomorphism of lattices $\phi : L \rightarrow \mathbf{2}$
such that $\phi(x) = 1$ and $\phi(y) = 0$.}
\end{quote}
Indeed, by way of Theorem \ref{hilbert} for the Boolean algebra $\mathbf{2}$,
we have $x \leq y$ if and only if $(x,1), (y,0) \vdash$.
Therefore, if $x \nleq y$,
then, by CT and classical logic,
there is an ideal element
witnessing $(x,1), (y,0) \nvdash$.
This is a homomorphism $\phi : L \rightarrow \mathbf{2}$ of lattices
such that $\phi(x) = 1$ and $\phi(y) = 0$.
\end{rem}

\subsection{Extension by conservation}\label{section:extbycon}

Now let us see how Proposition \ref{hilbert} relates to the classical extension theorem.
Suppose that $L$ and $L'$ are distributive lattices.
Given a finite discrete Boolean algebra $B$,
we have two entailment relations as above,
which we denote by $\vdash$ and $\vdash'$, respectively,
each of which describes lattice maps $L \rightarrow B$.
Notice that every lattice map $\phi : L \rightarrow L'$ gives way to an \emph{interpretation} (cf. Section \ref{sec:er})
\[
f_\phi : (L \times B, \vdash) \rightarrow (L' \times B, \vdash'),
\quad
(x, a) \mapsto (\phi(x), a).
\]
Indeed, it suffices to show that $f_\phi$ maps generating axioms for $\vdash$ to those of $\vdash'$,
which is clear since $\phi$ preserves the lattice structure.

\begin{prop}
If $\phi : L \rightarrow L'$ is an injective map of lattices and $B$ a finite discrete Boolean algebra,
then the induced interpretation
\[
f_\phi : (L \times B, \vdash) \rightarrow (L' \times B, \vdash'),
\quad
(x, a) \mapsto (\phi(x), a)
\]
is conservative, i.e., $f_\phi(X) \vdash' f_\phi(Y)$ implies $X \vdash Y$.
\end{prop}

%
%
%

\begin{proof}
It suffices to show conservation of inconsistent sets.
Hence, let $X \subseteq L \times B$ and suppose that $f_\phi(X) \vdash'$.
According to Theorem \ref{hilbert},
there is $e \in \mathrm{At}\, B$ such that
\[
\bigwedge_{a \geq e} \bigwedge {f_\phi(X)}_a
\leq
\bigvee_{a \geq e} \bigvee {f_\phi(X)}_{- a}
\]
in $L'$.
This means
\[
\phi \Big( \bigwedge_{a \geq e} \bigwedge X_a \Big)
\leq
\phi \Big( \bigvee_{a \geq e} \bigvee X_{- a} \Big)
\]
and implies $X \vdash$ by injectivity and Theorem \ref{hilbert}, once more.
\end{proof}

\begin{rem}
If $\phi : L \rightarrow L'$ is a lattice map for which the induced
interpretation $f_\phi$ is conservative with regard to a finite discrete Boolean algebra $B$,
then $\phi$ is injective.
%
In fact, recall from Example \ref{exa:orderreflection} that we have
\[
\phi(x) \leq \phi(y)
\qquad
\text{if and only if}
\qquad
(\phi(x),1) \vdash' (\phi(y),1),
\]
which is to say that $\phi(x) \leq \phi(y)$ if and only if $f_\phi(x,1) \vdash' f_\phi(y,1)$.
Likewise, $x \leq y$ is equivalent to having the entailment $(x,1) \vdash (y,1)$.
Therefore, if $f_\phi$ is conservative,
then $\phi$ is injective.
\end{rem}

Since $f_\phi$ is an interpretation,
the inverse image mapping of $f_\phi$
restricts on ideal elements
\[
f_\phi^{-1} : \mathfrak{Spec}(\vdash') \rightarrow \mathfrak{Spec}(\vdash)
\]
and it is easy to see that $f_\phi^{-1}(\alpha) = \alpha \circ \phi$.
Recall from Section \ref{sec:er},
that---with completeness at hand---conservation amounts to $f_\phi^{-1}$ being surjective.
Thus, if $\phi : L \rightarrow L'$ is a monomorphism of lattices,
then for every $\alpha: L \rightarrow B$
there is $\beta : L' \rightarrow B$ with $\alpha = \beta \circ \phi$.
In other words, lattice maps $L \rightarrow B$ extend along embeddings:
\[
\begin{tikzcd}
L \arrow[]{r}{\phi} \arrow{dr}[swap]{\forall\alpha} &		L' \arrow[dashed]{d}{\exists\beta}	\\
~ &		B
\end{tikzcd}
\]
\begin{cor}[\textbf{ZFC}]\label{finBAinj}
Every finite Boolean algebra is injective in the category of distributive lattices.
\end{cor}

Notice that if $L$ is a sublattice of $L'$ and if $\phi$ denotes inclusion of the former,
then $f_\phi^{-1}$ is nothing but the restriction of lattice maps to the sublattice,
and the extension is conservative if and only if restriction is surjective.
But we have to emphasize again that this requires completeness!

It is well known \cite{balbes:proin,banaschewski:inhulls} that every injective distributive lattice is a Boolean algebra.
In Section \ref{section:constocomp} we will see that complements are necessary for conservation:
a finite distributive lattice which lacks a complement for at least one of its elements
cannot allow for a result analogous to Theorem \ref{hilbert},
and cannot be injective among distributive lattices.

\section{Notes on injective Heyting algebras}

A (bounded) lattice $L$ is said to be a \emph{Heyting algebra} if,
for every pair of elements $x,y \in L$,
there is an element $x \rightarrow y \in L$ such that,
for every $z \in L$,
\[
z \leq x \rightarrow y
\qquad
\text{if and only if}
\qquad
z \wedge x \leq y.
\]
It is well-known that any Heyting algebra is distributive \cite{johnstone:stsp}.
A homomorphism of Heyting algebras is a lattice homomorphism that preserves \emph{implication} ($\rightarrow$).
Every Boolean algebra $B$ is a Heyting algebra with $x \rightarrow y \equiv -x \vee y$.
If $L$ is a Heyting algebra, \emph{negation} is defined by $\neg x \equiv x \rightarrow 0$.
A Heyting algebra $L$ is a Boolean algebra if and only if $\neg \neg x = x$ for every $x \in L$.
An element $x \in L$ is said to be \emph{regular} if $\neg \neg x = x$.
The set $L_{\neg\neg}$ of all regular elements of $L$
with the induced order is a Boolean algebra:
it is a sub-meet-semilattice of $L$, with joins defined by $x \vee_{L_{\neg\neg}} y \equiv \neg\neg( x \vee y )$.
We refer to \cite{johnstone:stsp}.
It is well known that
\emph{Booleanization}
\[
\neg \neg : L \rightarrow L_{\neg\neg},
\quad
x \mapsto \neg\neg x
\]
is a homomorphism of Heyting algebras \cite{balbes:injheyt}.

As shown by Balbes and Horn \cite{balbes:injheyt}, a Heyting algebra is injective (in the category of Heyting algebras)
if and only if it is a complete Boolean algebra.
The proof of this result employs Sikorski's theorem and argues with the Boolean algebra of regular elements
of a Heyting algebra.
We adopt the idea and consider the corresponding conservation result with regard to finite
discrete Boolean algebras.

To this end, let $L$ be a Heyting algebra and $B$ a finite discrete Boolean algebra.
The entailment relation of Heyting algebra morphisms $L \rightarrow B$
is inductively generated
by the set of all instances of the following axioms.
\begin{align*}
(x, a), (x, b) &\vdash \tag{s}\\
(x, a), (y, b) &\vdash (x \wedge y, a \wedge b) \tag{$\wedge$}\\
(x, a), (y, b) &\vdash (x \vee y, a \vee b) \tag{$\vee$}\\
(x, a), (y, b) &\vdash (x \rightarrow y, a \rightarrow b) \tag{$\rightarrow$}\\
&\vdash (0_L, 0_B) \tag{$0$}\\
&\vdash (1_L, 1_B) \tag{$1$}\\
&\vdash \Set{ (x,a) : a \in B} \tag{t}
\end{align*}
where $a \neq b$ in (s).

Thus $\vdash$ generates from the entailment relation of lattice maps $L \rightarrow B$
by adjoining additional axioms for implication.

Even though $L$ need not be Boolean itself,
the value of an arbitrary element $x \in L$ under a Heyting algebra homomorphism
$L \rightarrow B$ is determined by the value of $\neg \neg x$, and vice versa:

\begin{lem}\label{lem:dneg}
For every $(x,a) \in L \times B$ we have $(x,a) \dashv\vdash (\neg\neg x, a)$.
\end{lem}

\begin{proof}
Both $(x,a), (0_L,0_B) \vdash (\neg x, \neg a)$ and $(\neg x, \neg a), (0_L, 0_B) \vdash (\neg \neg x, \neg \neg a)$
are instances of axiom $(\rightarrow)$.
By cut with $\vdash (0_L, 0_B)$, and since $\neg \neg a = a$ in $B$,
we get $(x,a) \vdash (\neg \neg x, a)$.
Once we have this entailment,
it follows that $(x,a), (\neg \neg x, b) \vdash$ for every $b \neq a$.
Therefore, we may cut (t) for $x$ accordingly,
and infer the converse entailment, too.
\end{proof}

Now let $\vdash'$ denote the entailment relation of lattice maps $L_{\neg\neg} \rightarrow B$,
generated as before without the axiom for implication.
Notice that double negation induces an interpretation
\[
(L \times B, \vdash) \rightarrow (L_{\neg\neg} \times B, \vdash),
\quad
(x,a) \mapsto (\neg\neg x, a).
\]
In the other direction we have an inclusion $L_{\neg \neg} \hookrightarrow L$
that preserves meets but not in general joins.
However, an axiom of the form
\[
(x,a), (y,b) \vdash' (x \vee_{L_{\neg\neg}} y, a \vee b)
\]
means
\[
(x,a), (y,b) \vdash' (\neg\neg(x \vee y), a \vee b)
\]
which can be inferred also with regard to $\vdash$,
by way of axiom $(\vee)$ and in view of Lemma \ref{lem:dneg}.
It follows that we have a conservative interpretation of entailment relations.
Employing Theorem \ref{hilbert},
we get the formal Nullstellensatz for this entailment relation.
Here is how to describe inconsistent sets explicitly:

\begin{cor}
For every finite subset $X$ of $L \times B$,
the following are equivalent.
\begin{enumerate}
\item
$X \vdash$
\item
There is an atom $e \in \mathrm{At}\,B$ such that
\[
\bigwedge_{a \geq e} \bigwedge \neg\neg X_a \leq
\bigvee_{a \geq e} \bigvee \neg\neg X_{-a}.
\]
\end{enumerate}
\end{cor}

The general Nullstellensatz for $\vdash$ derives from the description of inconsistent sets.
Conservation is an immediate consequence.
\section{From conservation to complements}\label{section:constocomp}

At the outset,
the way in which we have generated the entailment relation in Section \ref{sequentsformaps}
did not depend on the structure of $B$ as a Boolean algebra,
and might as well be carried out with any finite lattice $D$ instead.
At least the ideal elements would exactly be the lattice maps $L \rightarrow D$.
One might thus be tempted to question whether and to what extent
complements in $D$ are necessary at all
in order to allow for a corresponding conservation result.
Incidentally,
the entailment relation for $\mathbf{2}$-valued maps has an important application,
demonstrated in \cite[Theorem 11]{coquand:erdl}, which may be used to resolve this question:

\begin{prop}\label{genBA}
If $B, i : L \times \mathbf{2} \rightarrow B$ is the distributive lattice generated by $(L\times \mathbf{2}, \vdash)$,
then $B$ is a Boolean algebra and $L$ embeds in $B$.
\end{prop}

Now let us say that a finite distributive lattice $D$ is \emph{conservative} in case the following holds:
if $L$ and $L'$ are distributive lattices,
$L$ being a sublattice of $L'$,
then $(L \times D, \vdash) \hookrightarrow (L' \times D, \vdash')$
is a conservative extension of entailment relations,
where $\vdash$ and $\vdash'$ are generated as in \ref{sequentsformaps},
with $D$ in place of $B$,
respectively.

\begin{prop}[\textbf{ZFC}]\label{consisbool}
For every finite distributive lattice $D$,
the following are equivalent.
\begin{enumerate}
\item
$D$ is complemented.
\item
$D$ is conservative.
\item
$D$ is injective.
\end{enumerate}
\end{prop}

\begin{proof}
We have already seen that every finite (discrete)
Boolean algebra is conservative in the sense specified before,
and injectivity is a classical consequence of completeness.
On the other hand, let $D$ be a finite distributive lattice and suppose that it is injective among distributive lattices.
This $D$ can be considered a sublattice of a Boolean algebra $B$,
applying, for instance, Proposition \ref{genBA}.
\[
\begin{tikzcd}
D \arrow[hook]{r}\arrow{rd}[swap]{\mathrm{id}_D}		&		B\arrow[dashed]{d}{\exists f}		\\
~		&		D
\end{tikzcd}
\]
By way of injectivity,
it follows that $D$ is the homomorphic image of a Boolean algebra,
whence Boolean itself.
\end{proof}

What goes wrong in case $D$ is not Boolean?
Towards an answer, Proposition \ref{consisbool} might not be considered all too helpful,
taking into account that its proof invokes the completeness theorem.
There is another,
more concrete argument,
which provides an explicit counterexample to conservation.

Once again, let $D$ be a finite distributive lattice, 
and suppose that $d_0 \in D$ is not complemented.
We consider the lattice $\mathbf{2}^2 = \set{ (0,0), (0,1), (1,0), (1,1) }$.
Let $\vdash\, \subseteq \mathbf{2}^2 \times D$
be the entailment relation of $D$-valued lattice maps on $\mathbf{2}^2$,
generated by axioms as in Section \ref{sequentsformaps},
with $D$ in place of $B$.
Let
\[
X = \set{ ((0,0), 0_D), ((1,1), 1_D), ((0,1),d_0) }.
\]
Notice that for every $d \in D$ we have $X, ((1,0),d) \vdash$.
This is because $d_0$ is supposed not have a complement,
and therefore, for any $d \in D$,
either we have $d_0 \wedge d \neq 0$ or $d_0 \vee d \neq 1$.
Involving appropriate instances of $(\wedge)$ or $(\vee)$
as well as corresponding instances of (s),
we infer that for every $d \in D$ the set
$X, ((1,0),d)$ is inconsistent with respect to $\vdash$.
Then we instantiate (t), which reads
\[
\vdash \Set{ ((1,0),d) : d \in D },
\]
and by way of cut we get $X \vdash$.
However, this set $X$ is \emph{not} inconsistent for the entailment relation of $D$-valued maps
on the sublattice $\set{ (0,0), (0,1), (1,1) }$.
In fact, for this very entailment relation $X$ is an ideal element!
We conclude that a finite distributive lattice,
which lacks a complement for at least one of its elements,
cannot be conservative either.

%
%

\section{Conclusion}

For every distributive lattice it is possible to generate an entailment relation
the ideal elements of which are precisely the lattice maps on $L$ with values
in a given finite discrete Boolean algebra.
In this manner,
the assertion that any such map extends onto ambient lattices---itself
classically equivalent to the Boolean prime ideal theorem---gets logically described and turns into
a syntactical conservation theorem.

Conservation, in turn,
has an elementary constructive proof,
whence can be considered a constructive version of the classical extension theorem.
As put by Coquand and Lombardi \cite{coq:hidden-krull},
to have a constructive version of a classical theorem means
to have ``\emph{a theorem the proof of which is constructive,
which has a clear computational content, and from which we can recover the usual version
of the abstract theorem by an immediate application of a well classified non-constructive principle}.''
In our case, the non-constructive principle in question is a suitable form of the completeness theorem.


Sikorski's extension theorem undoubtedly has interesting consequences in classical mathematics, e.g.,
by duality theory it leads over to Gleason's theorem,
characterizing the projective objects in the category of compact Hausdorff spaces
precisely as the extremally disconnected spaces \cite[Theorem 3.7]{johnstone:stsp}.
Yet our conservation result begs the question:
does it allow to substitute applications
of its classical counterpart in a manner that maintains computational information?
 -- We do not know at this point.
The situation appears rather reminiscent of the Hahn-Banach theorem:
whether one can make computational use of the Hahn-Banach theorem itself seems not clear either \cite{coquand:dlhb}.

\section*{Acknowledgements}

The research that has led to this paper
was carried out within the project ``Categorical localisation: methods and foundations'' (CATLOC)
funded by the University of Verona within the programme ``Ricerca di Base 2015'';
the related financial support is gratefully acknowledged.
The final version of this paper was prepared within the project
``A New Dawn of Intuitionism: Mathematical and Philosophical Advances'' (ID 60842)
funded by the John Templeton Foundation,
as well as within the project
``Dipartimenti di Eccellenza 2018-2022'' of the Italian Ministry of Education, Universities and Research (MIUR).\footnote{The opinions expressed in this publication are those of the authors and
do not necessarily reflect the views of the John Templeton Foundation.}
The authors are most thankful to the anonymous referees for their help- and insightful constructive critique
which has led to an improved presentation of this paper.
The authors would further like to express their sincere gratitude to Peter Schuster
for his advice, encouragement, and many discussions on the subject matter.


\begin{thebibliography}{10}

\bibitem{aczel:notes}
Peter Aczel and Michael Rathjen.
\newblock Notes on constructive set theory.
\newblock Technical report, Institut Mittag--Leffler, 2000/01.
\newblock Report No. 40.

\bibitem{aczel:cstdraft}
Peter Aczel and Michael Rathjen.
\newblock {Constructive set theory}.
\newblock Book draft, 2010.

\bibitem{balbes:proin}
Raymond Balbes.
\newblock Projective and injective distributive lattices.
\newblock {\em Pacific J.~Math.}, 21(3):405--420, 1967.

\bibitem{balbes:injheyt}
Raymond Balbes and Alfred Horn.
\newblock Injective and projective {Heyting} algebras.
\newblock {\em Trans.~Amer.~Math.~Soc.},
  148(549--559), 1970.

\bibitem{banaschewski:inhulls}
Bernhard Banaschewski and G\"unter Bruns.
\newblock Injective hulls in the category of distributive lattices.
\newblock {\em {J.~Reine Angew.~Math.}},
  232:102--109, 1968.

\bibitem{bell:strength}
{John L.} Bell.
\newblock On the strength of the {Sikorski} extension theorem for {Boolean}
  algebras.
\newblock {\em J.~Symbolic Logic}, 48(3):841--846, 1983.

\bibitem{bell:pers}
{John L.} Bell.
\newblock Zorn's lemma and complete {Boolean} algebras in intuitionistic type
  theories.
\newblock {\em J.~Symbolic Logic}, 62(4):1265--1279, 1997.

\bibitem{bell:ist}
{John L.} Bell.
\newblock {\em Intuitionistic Set Theory}, volume~50 of {\em Studies in Logic}.
\newblock College Publications, London, 2014.

\bibitem{coquand:erdl}
Jan {Cederquist} and Thierry {Coquand}.
\newblock Entailment relations and distributive lattices.
\newblock In Samuel~R. Buss, Petr H{\'a}jek, and Pavel Pudl{\'a}k, editors,
  {\em Logic Colloquium '98: Proceedings of the Annual European Summer Meeting
  of the Association for Symbolic Logic}, volume~13 of {\em Lecture Notes in
  Logic}, pages 110--123. AK Peters/Springer, 2000.

\bibitem{negri:hbintt}
Jan Cederquist, Thierry Coquand, and Sara Negri.
\newblock {The Hahn-Banach Theorem in Type Theory}.
\newblock In Giovanni Sambin and {Jan M.} Smith, editors, {\em Twenty-Five
  Years of Constructive Type Theory}. Oxford University Press, Oxford, 1998.

\bibitem{coquand:dlhb}
Thierry Coquand.
\newblock A direct proof of the localic {Hahn-Banach} theorem, 2000.
\newblock URL: \url{http://www.cse.chalmers.se/~coquand/formal.html}.

\bibitem{coquand:topandsc}
Thierry Coquand.
\newblock Topology and sequent calculus.
\newblock Conference presentation, June 2000.
\newblock Topology in Computer Science: Constructivity; Asymmetry and
  Partiality; Digitization.
\newblock URL: \url{http://www.cse.chalmers.se/~coquand/formal.html}.

\bibitem{coquand:geomhb}
Thierry Coquand.
\newblock {Geometric Hahn-Banach theorem}.
\newblock {\em Math.~Proc.~Cambridge Philos.~Soc.}, 140(2):313--315, 2006.

\bibitem{coq:hidden-krull}
Thierry Coquand and Henri Lombardi.
\newblock Hidden constructions in abstract algebra: Krull dimension of
  distributive lattices and commutative rings.
\newblock In M.~Fontana, S.-E. Kabbaj, and S.~Wiegand, editors, {\em
  Commutative {R}ing {T}heory and {A}pplications}, volume 231 of {\em
  Lect.~{N}otes {P}ure {A}ppl.~{M}athematics}, pages 477--499, Reading, MA,
  2002. Addison-Wesley.

\bibitem{coquand:vdp}
Thierry Coquand and Henrik Persson.
\newblock Valuations and {D}edekind's {P}rague theorem.
\newblock {\em J.~Pure Appl.~Algebra}, 155:121--129, 2001.

\bibitem{coquand:sfc}
Thierry Coquand and {Guo-Qiang} Zhang.
\newblock Sequents, frames, and completeness.
\newblock In Helmut Schwichtenberg and {Peter G.} Clote, editors, {\em Computer
  Science Logic. 14th International Workshop, CSL 2000 Annual Conference of the
  EACSL}.

\bibitem{cos:dyn}
Michel Coste, Henri Lombardi, and Marie-Fran\c{c}oise Roy.
\newblock {Dynamical method in algebra: Effective Nullstellens\"atze.}
\newblock {\em Ann.~Pure Appl.~Logic}, 111(3):203--256, 2001.

\bibitem{halmos:loba}
{Paul R.} Halmos.
\newblock {\em Lectures on Boolean Algebras}.
\newblock Van Nostrand, New York, 1963.

\bibitem{johnstone:stsp}
Peter Johnstone.
\newblock {\em Stone Spaces}.
\newblock Cambridge University Press, Cambridge, 1982.

\bibitem{handbook:BA1}
Sabine Koppelberg.
\newblock {\em Handbook of Boolean Algebras}, volume~1.
\newblock North-Holland, Amsterdam, 1989.

\bibitem{lifschitz:sct}
Vladimir Lifschitz.
\newblock Semantical completeness theorems in logic and algebra.
\newblock {\em Proc.~Amer.~Math.~Soc.}, 79(89--96),
  1980.

\bibitem{lombardi:krull}
Henri Lombardi.
\newblock Dimension de {K}rull, {N}ullstellens\"atze et \'evaluation dynamique.
\newblock {\em Math.~Zeitschrift}, 242:23--46, 2002.

\bibitem{lombardiquitte:constructive}
Henri Lombardi and Claude Quitt\'e.
\newblock {\em Commutative Algebra: Constructive Methods}.
\newblock Springer Netherlands, Dordrecht, 2015.

\bibitem{luxemburg:sik}
{W. A. J.} Luxemburg.
\newblock A remark on {Sikorski's} extension theorem for homomorphisms in the
  theory of {Boolean} algebras.
\newblock {\em Fund.~Math.}, 55(2):239--247, 1964.

\bibitem{mulvey:globhb}
{Christopher J.} Mulvey and Joan {Wick-Pelletier}.
\newblock {A globalization of the Hahn-Banach theorem}.
\newblock {\em Adv.~Math.}, 89:1--59, 1991.

\bibitem{negri:ptorder}
Sara Negri, Jan {von Plato}, and Thierry Coquand.
\newblock {Proof-theoretical analysis of order relations}.
\newblock {\em Arch.~Math.~Logic}, 43:297--309, 2004.

\bibitem{rsw:edde:abstract}
Davide Rinaldi, Peter Schuster, and Daniel Wessel.
\newblock Eliminating disjunctions by disjunction elimination.
\newblock {\em Bull.~Symbolic Logic}, 23(2):181--200, 2017.

\bibitem{rsw:edde}
Davide Rinaldi, Peter Schuster, and Daniel Wessel.
\newblock Eliminating disjunctions by disjunction elimination.
\newblock {\em Indag.~Math.~(N.S.)}, 29(1):226--259, 2018.

\bibitem{rinaldiwessel:cut}
Davide Rinaldi and Daniel Wessel.
\newblock Cut elimination for entailment relations.
\newblock Subm.~manuscript, 2018.

\bibitem{rubin:equivalents}
Herman Rubin and {Jean E.} Rubin.
\newblock {\em {Equivalents of the Axiom of Choice, II}}, volume 116 of {\em
  Studies in Logic and the Foundations of Mathematics}.
\newblock North-Holland, Amsterdam, 1985.

\bibitem{sambin:somepoints}
Giovanni Sambin.
\newblock Some points in formal topology.
\newblock {\em Theoret.~Comput.~Sci.}, 305(1-3):347--408, 2003.

\bibitem{scarp:meta}
Bruno Scarpellini.
\newblock On the metamathematics of rings and integral domains.
\newblock {\em Trans.~Amer.~Math.~Soc.}, 138(71--96),
  1969.

\bibitem{scott:engender}
Dana Scott.
\newblock {On engendering an illusion of understanding}.
\newblock {\em J.~Philos.}, 68(21):787--807, 1971.

\bibitem{sco:com}
Dana Scott.
\newblock Completeness and axiomatizability in many-valued logic.
\newblock In Leon Henkin, John Addison, {C.C.} Chang, William Craig, Dana
  Scott, and Robert Vaught, editors, {\em Proceedings of the {T}arski
  {S}ymposium ({P}roc. {S}ympos. {P}ure {M}ath., {V}ol. {XXV}, {U}niv.
  {C}alifornia, {B}erkeley, {C}alif., 1971)}, pages 411--435. Amer. Math. Soc.,
  Providence, RI, 1974.

\bibitem{sco:bac}
Dana~S. Scott.
\newblock Background to formalization.
\newblock In Hugues Leblanc, editor, {\em Truth, syntax and modality ({P}roc.
  {C}onf. {A}lternative {S}emantics, {T}emple {U}niv., {P}hiladelphia, {P}a.,
  1970)}, pages 244--273. Studies in Logic and the Foundations of Math., Vol.
  68. North-Holland, Amsterdam, 1973.

\bibitem{sikorski:ba}
Roman Sikorski.
\newblock {\em Boolean Algebras}.
\newblock {Ergebnisse der Mathematik und ihrer Grenzgebiete. Band 25}.
  Springer-Verlag, Berlin, 1969.

\end{thebibliography}
\end{document}